\documentclass[10pt,conference,a4paper]{IEEEtran}
\usepackage{fancyhdr} 
\usepackage[USenglish,american]{babel}
\usepackage{epsfig,graphics,subfigure,graphicx,latexsym,longtable,amsmath,amscd,latexsym,amssymb,mathrsfs,syntonly,eucal}
\usepackage{multirow}
\usepackage[usenames]{color}
\usepackage[T1]{fontenc}
\usepackage{bm,cite}
\usepackage{amsbsy}
\usepackage{latexsym}
\usepackage{wasysym}
\usepackage{placeins}
\usepackage[lined,boxed,commentsnumbered]{algorithm2e}
\usepackage{lipsum}
\usepackage{url}
\usepackage{array}
\usepackage{tabu}
\SetKwInput{KwInput}{Input}
\SetKwInput{KwOutput}{Output}
\usepackage{booktabs,siunitx}
\usepackage{booktabs}
\usepackage{makecell}

\usepackage[a4paper,
        bindingoffset=0.05in,
        left=14mm,
        right=14mm,
        top=1.9cm,
        bottom=4.3cm,
        ]{geometry} 

 \newcommand\blfootnote[1]{%
  \begingroup
  \renewcommand\thefootnote{}\footnote{#1}%
  \addtocounter{footnote}{-1}%
  \endgroup
}

\usepackage[acronyms,nonumberlist,nopostdot,nomain,nogroupskip,acronymlists={hidden}]{glossaries}
\newglossary[algh]{hidden}{acrh}{acnh}{Hidden Acronyms}
\glsdisablehyper
\newacronym{3gpp}{3GPP}{3rd Generation Partnership Project}
\newacronym{4g}{4G}{4th generation}
\newacronym{5g}{5G}{5th generation}
\newacronym{6g}{6G}{6th generation}
\newacronym{5gc}{5GC}{5G Core}
\newacronym{adc}{ADC}{Analog to Digital Converter}
\newacronym{aerpaw}{AERPAW}{Aerial Experimentation and Research Platform for Advanced Wireless}
\newacronym{ai}{AI}{Artificial Intelligence}
\newacronym{aimd}{AIMD}{Additive Increase Multiplicative Decrease}
\newacronym{am}{AM}{Acknowledged Mode}
\newacronym{amc}{AMC}{Adaptive Modulation and Coding}
\newacronym{amf}{AMF}{Access and Mobility Management Function}
\newacronym{aoa}{AoA}{Angle of Arrival}
\newacronym{aops}{AOPS}{Adaptive Order Prediction Scheduling}
\newacronym{api}{API}{Application Programming Interface}
\newacronym{apn}{APN}{Access Point Name}
\newacronym{aqm}{AQM}{Active Queue Management}
\newacronym{ausf}{AUSF}{Authentication Server Function}
\newacronym{avc}{AVC}{Advanced Video Coding}
\newacronym{awgn}{AGWN}{Additive White Gaussian Noise}
\newacronym{balia}{BALIA}{Balanced Link Adaptation Algorithm}
\newacronym{bbu}{BBU}{Base Band Unit}
\newacronym{bdp}{BDP}{Bandwidth-Delay Product}
\newacronym{ber}{BER}{Bit Error Rate}
\newacronym{bf}{BF}{Beamforming}
\newacronym{bler}{BLER}{Block Error Rate}
\newacronym{brr}{BRR}{Bayesian Ridge Regressor}
\newacronym{bsr}{BSR}{Buffer Status Report}
\newacronym{bs}{BS}{Base Station}
\newacronym{bss}{BSS}{Business Support System}
\newacronym{ca}{CA}{Carrier Aggregation}
\newacronym{caas}{CaaS}{Connectivity-as-a-Service}
\newacronym{cb}{CB}{Code Block}
\newacronym{cc}{CC}{Congestion Control}
\newacronym{ccid}{CCID}{Congestion Control ID}
\newacronym{cco}{CC}{Carrier Component}
\newacronym{cdd}{CDD}{Cyclic Delay Diversity}
\newacronym{cdf}{CDF}{Cumulative Distribution Function}
\newacronym{cdn}{CDN}{Content Distribution Network}
\newacronym{csi-rs}{CSI-RS}{Channel State Information Reference Signal}
\newacronym{cir}{CIR}{Channel Impulse Response}
\newacronym{cn}{CN}{Core Network}
\newacronym{codel}{CoDel}{Controlled Delay Management}
\newacronym{comac}{COMAC}{Converged Multi-Access and Core}
\newacronym{cord}{CORD}{Central Office Re-architected as a Datacenter}
\newacronym{cornet}{CORNET}{COgnitive Radio NETwork}
\newacronym{cosmos}{COSMOS}{Cloud Enhanced Open Software Defined Mobile Wireless Testbed for City-Scale Deployment}
\newacronym{cots}{COTS}{Commercial Off-the-Shelf}
\newacronym{cp}{CP}{Cyclic Prefix}
\newacronym{cpu}{CPU}{Central Processing Unit}
\newacronym{cqi}{CQI}{Channel Quality Information}
\newacronym{cr}{CR}{Cognitive Radio}
\newacronym{cran}{CRAN}{Cloud \gls{ran}}
\newacronym{crc}{CRC}{Cyclic Redundancy Check}
\newacronym{crlb}{CRLB}{Cramér–Rao Lower Bound}
\newacronym{crs}{CRS}{Cell Reference Signal}
\newacronym{csi}{CSI}{Channel State Information}
\newacronym{csirs}{CSI-RS}{Channel State Information - Reference Signal}
\newacronym{cu}{CU}{Central Unit}
\newacronym{d2tcp}{D$^2$TCP}{Deadline-aware Data center TCP}
\newacronym{d3}{D$^3$}{Deadline-Driven Delivery}
\newacronym{dac}{DAC}{Digital to Analog Converter}
\newacronym{dag}{DAG}{Directed Acyclic Graph}
\newacronym{darpa}{DARPA}{Defense Advanced Research Projects Agency}
\newacronym{das}{DAS}{Distributed Antenna System}
\newacronym{dash}{DASH}{Dynamic Adaptive Streaming over HTTP}
\newacronym{dc}{DC}{Dual Connectivity}
\newacronym{dccp}{DCCP}{Datagram Congestion Control Protocol}
\newacronym{dce}{DCE}{Direct Code Execution}
\newacronym{dci}{DCI}{Downlink Control Information}
\newacronym{dcl}{DCL}{Dear Colleague Letter}
\newacronym{dctcp}{DCTCP}{Data Center TCP}
\newacronym{dl}{DL}{Downlink}
\newacronym{dmr}{DMR}{Deadline Miss Ratio}
\newacronym{dmrs}{DMRS}{DeModulation Reference Signal}
\newacronym{drlcc}{DRL-CC}{Deep Reinforcement Learning Congestion Control}
\newacronym{drs}{DRS}{Discovery Reference Signal}
\newacronym{du}{DU}{Distributed Unit}
\newacronym{e2e}{E2E}{end-to-end}
\newacronym{ecaas}{ECaaS}{Edge-Cloud-as-a-Service}
\newacronym{ecn}{ECN}{Explicit Congestion Notification}
\newacronym{edf}{EDF}{Earliest Deadline First}
\newacronym{embb}{eMBB}{Enhanced Mobile Broadband}
\newacronym{empower}{EMPOWER}{EMpowering transatlantic PlatfOrms for advanced WirEless Research}
\newacronym{enb}{eNB}{evolved Node Base}
\newacronym{endc}{EN-DC}{E-UTRAN-\gls{nr} \gls{dc}}
\newacronym{epc}{EPC}{Evolved Packet Core}
\newacronym{eps}{EPS}{Evolved Packet System}
\newacronym{es}{ES}{Edge Server}
\newacronym{etsi}{ETSI}{European Telecommunications Standards Institute}
\newacronym[firstplural=Estimated Times of Arrival (ETAs)]{eta}{ETA}{Estimated Time of Arrival}
\newacronym{eutran}{E-UTRAN}{Evolved Universal Terrestrial Access Network}
\newacronym{faas}{FaaS}{Function-as-a-Service}
\newacronym{fapi}{FAPI}{Functional Application Platform Interface}
\newacronym{fcc}{FCC}{Federal Communications Commission}
\newacronym{fdd}{FDD}{Frequency Division Duplexing}
\newacronym{fdm}{FDM}{Frequency Division Multiplexing}
\newacronym{fdma}{FDMA}{Frequency Division Multiple Access}
\newacronym{fed4fire}{FED4FIRE+}{Federation 4 Future Internet Research and Experimentation Plus}
\newacronym{fir}{FIR}{Finite Impulse Response}
\newacronym{fit}{FIT}{Future \acrlong{iot}}
\newacronym{fmcw}{FMCW}{Frequency-Modulated Continuous-Wave}  
\newacronym{fft}{FFT}{Fourier Transform}
\newacronym{fpga}{FPGA}{Field Programmable Gate Array}
\newacronym{fr2}{FR2}{Frequency Range 2}
\newacronym{fs}{FS}{Fast Switching}
\newacronym{fscc}{FSCC}{Flow Sharing Congestion Control}
\newacronym{ftp}{FTP}{File Transfer Protocol}
\newacronym{fw}{FW}{Flow Window}
\newacronym{ge}{GE}{Gaussian Elimination}
\newacronym{gnb}{gNB}{Base Station}
\newacronym{gop}{GOP}{Group of Pictures}
\newacronym{gpr}{GPR}{Gaussian Process Regressor}
\newacronym{gps}{GPS}{Global Positioning System}
\newacronym{gpu}{GPU}{Graphics Processing Unit}
\newacronym{gtp}{GTP}{GPRS Tunneling Protocol}
\newacronym{gtpc}{GTP-C}{GPRS Tunnelling Protocol Control Plane}
\newacronym{gtpu}{GTP-U}{GPRS Tunnelling Protocol User Plane}
\newacronym{gtpv2c}{GTPv2-C}{\gls{gtp} v2 - Control}
\newacronym{gw}{GW}{Gateway}
\newacronym{harq}{HARQ}{Hybrid Automatic Repeat reQuest}
\newacronym{hetnet}{HetNet}{Heterogeneous Network}
\newacronym{hh}{HH}{Hard Handover}
\newacronym{hol}{HOL}{Head-of-Line}
\newacronym{hqf}{HQF}{Highest-quality-first}
\newacronym{hss}{HSS}{Home Subscription Server}
\newacronym{http}{HTTP}{HyperText Transfer Protocol}
\newacronym{ia}{IA}{Initial Access}
\newacronym{iab}{IAB}{Integrated Access and Backhaul}
\newacronym{ic}{IC}{Incident Command}
\newacronym{ietf}{IETF}{Internet Engineering Task Force}
\newacronym{imsi}{IMSI}{International Mobile Subscriber Identity}
\newacronym{imt}{IMT}{International Mobile Telecommunication}
\newacronym{iot}{IoT}{Internet of Things}
\newacronym{ip}{IP}{Internet Protocol}
\newacronym{isac}{ISaC}{Integrated Sensing and Communication}
\newacronym{isi}{ISI}{Intersymbol Interference}
\newacronym{itu}{ITU}{International Telecommunication Union}
\newacronym{kpi}{KPI}{Key Performance Indicator}
\newacronym{kvm}{KVM}{Kernel-based Virtual Machine}
\newacronym{ldpc}{LDPC}{Low-Density Parity Check}
\newacronym{los}{LoS}{Line-of-Sight}
\newacronym{lsm}{LSM}{Link-to-System Mapping}
\newacronym{lstm}{LSTM}{Long Short Term Memory}
\newacronym{lte}{LTE}{Long Term Evolution}
\newacronym{lxc}{LXC}{Linux Container}
\newacronym{m2m}{M2M}{Machine to Machine}
\newacronym{mac}{MAC}{Medium Access Control}
\newacronym{manet}{MANET}{Mobile Ad Hoc Network}
\newacronym{mano}{MANO}{Management and Orchestration}
\newacronym{mc}{MC}{Multi-Connectivity}
\newacronym{mcc}{MCC}{Mobile Cloud Computing}
\newacronym{mchem}{MCHEM}{Massive Channel Emulator}
\newacronym{mcs}{MCS}{Modulation and Coding Scheme}
\newacronym{mcss}{MCS}{Modulation and Coding Schemes}
\newacronym{mec}{MEC}{Multi-access Edge Computing}
\newacronym{mec2}{MEC}{Mobile Edge Cloud}
\newacronym{mfc}{MFC}{Mobile Fog Computing}
\newacronym{mi}{MI}{Mutual Information}
\newacronym{mib}{MIB}{Master Information Block}
\newacronym{miesm}{MIESM}{Mutual Information Based Effective SINR}
\newacronym{mimo}{MIMO}{Multiple Input, Multiple Output}
\newacronym{mgen}{MGEN}{Multi-Generator}
\newacronym{ml}{ML}{Machine Learning}
\newacronym{mlr}{MLR}{Maximum-local-rate}
\newacronym[plural=\gls{mme}s,firstplural=Mobility Management Entities (MMEs)]{mme}{MME}{Mobility Management Entity}
\newacronym{mse}{MSE}{Mean Square Error}
\newacronym{mmtc}{mMTC}{Massive Machine-Type Communications}
\newacronym{mmwave}{mmWave}{millimeter wave}
\newacronym{mpdccp}{MP-DCCP}{Multipath Datagram Congestion Control Protocol}
\newacronym{mptcp}{MPTCP}{Multipath TCP}
\newacronym{mr}{MR}{Maximum Rate}
\newacronym{mrdc}{MR-DC}{Multi \gls{rat} \gls{dc}}
\newacronym{mss}{MSS}{Maximum Segment Size}
\newacronym{mt}{MT}{Mobile Termination}
\newacronym{mtd}{MTD}{Machine-Type Device}
\newacronym{mtu}{MTU}{Maximum Transmission Unit}
\newacronym{mumimo}{MU-MIMO}{Multi-user \gls{mimo}}
\newacronym{mvno}{MVNO}{Mobile Virtual Network Operator}
\newacronym{nalu}{NALU}{Network Abstraction Layer Unit}
\newacronym{nas}{NAS}{Network Attached Storage}
\newacronym{nbiot}{NB-IoT}{Narrow Band IoT}
\newacronym{nfv}{NFV}{Network Function Virtualization}
\newacronym{nfvi}{NFVI}{Network Function Virtualization Infrastructure}
\newacronym{nic}{NIC}{Network Interface Card}
\newacronym{nlos}{NLOS}{Non-Line-of-Sight}
\newacronym{now}{NOW}{Non Overlapping Window}
\newacronym{nrdz}{NRDZ}{National Radio Dynamic Zone}
\newacronym{nsf}{NSF}{National Science Foundation}
\newacronym{nsm}{NSM}{Network Service Mesh}
\newacronym[type=hidden]{nr}{NR}{New Radio}
\newacronym{nrf}{NRF}{Network Repository Function}
\newacronym{nsa}{NSA}{Non Stand Alone}
\newacronym{nse}{NSE}{Network Slicing Engine}
\newacronym{nssf}{NSSF}{Network Slice Selection Function}
\newacronym{ntp}{NTP}{Network Time Protocol}
\newacronym{o2i}{O2I}{Outdoor to Indoor}
\newacronym{oai}{OAI}{OpenAirInterface}
\newacronym{oaicn}{OAI-CN}{\gls{oai} \acrlong{cn}}
\newacronym{oairan}{OAI-RAN}{\acrlong{oai} \acrlong{ran}}
\newacronym{oam}{OAM}{Operations, Administration and Maintenance}
\newacronym{ofdm}{OFDM}{Orthogonal Frequency Division Multiplexing}
\newacronym{olia}{OLIA}{Opportunistic Linked Increase Algorithm}
\newacronym{omec}{OMEC}{Open Mobile Evolved Core}
\newacronym{onap}{ONAP}{Open Network Automation Platform}
\newacronym{onf}{ONF}{Open Networking Foundation}
\newacronym{onos}{ONOS}{Open Networking Operating System}
\newacronym{oom}{OOM}{\gls{onap} Operations Manager}
\newacronym{opnfv}{OPNFV}{Open Platform for \gls{nfv}}
\newacronym[type=hidden]{oran}{O-RAN}{Open \gls{ran}}
\newacronym{orbit}{ORBIT}{Open-Access Research Testbed for Next-Generation Wireless Networks}
\newacronym{os}{OS}{Operating System}
\newacronym{oss}{OSS}{Operations Support System}
\newacronym{otfs}{OTFS}{Orthogonal Time Frequency Space}
\newacronym{pa}{PA}{Position-aware}
\newacronym{pase}{PASE}{Prioritization, Arbitration, and Self-adjusting Endpoints}
\newacronym{pawr}{PAWR}{Platforms for Advanced Wireless Research}
\newacronym{pbch}{PBCH}{Physical Broadcast Channel}
\newacronym{pcef}{PCEF}{Policy and Charging Enforcement Function}
\newacronym{pcfich}{PCFICH}{Physical Control Format Indicator Channel}
\newacronym{pcrf}{PCRF}{Policy and Charging Rules Function}
\newacronym{pdcch}{PDCCH}{Physical Downlink Control Channel}
\newacronym{pdcp}{PDCP}{Packet Data Convergence Protocol}
\newacronym{pdsch}{PDSCH}{Physical Downlink Shared Channel}
\newacronym{pdu}{PDU}{Packet Data Unit}
\newacronym{pf}{PF}{Proportional Fair}
\newacronym{pgw}{PGW}{Packet Gateway}
\newacronym{phich}{PHICH}{Physical Hybrid ARQ Indicator Channel}
\newacronym{phy}{PHY}{Physical}
\newacronym{pmch}{PMCH}{Physical Multicast Channel}
\newacronym{pmi}{PMI}{Precoding Matrix Indicators}
\newacronym{powder}{POWDER}{Platform for Open Wireless Data-driven Experimental Research}
\newacronym{ppo}{PPO}{Proximal Policy Optimization}
\newacronym{ppp}{PPP}{Poisson Point Process}
\newacronym{prach}{PRACH}{Physical Random Access Channel}
\newacronym{prb}{PRB}{Physical Resource Block}
\newacronym{prs}{PRS}{Positioning Reference Signal}
\newacronym{psnr}{PSNR}{Peak Signal to Noise Ratio}
\newacronym{pss}{PSS}{Primary Synchronization Signal}
\newacronym{pucch}{PUCCH}{Physical Uplink Control Channel}
\newacronym{pusch}{PUSCH}{Physical Uplink Shared Channel}
\newacronym{qam}{QAM}{Quadrature Amplitude Modulation}
\newacronym{qci}{QCI}{\gls{qos} Class Identifier}
\newacronym{qoe}{QoE}{Quality of Experience}
\newacronym{qos}{QoS}{Quality of Service}
\newacronym{quic}{QUIC}{Quick UDP Internet Connections}
\newacronym{ra}{RA}{Random Access}
\newacronym{rach}{RACH}{Random Access Channel}
\newacronym{ran}{RAN}{Radio Access Network}
\newacronym{rar}{RAR}{Random Access Response}
\newacronym[firstplural=Radio Access Technologies (RATs)]{rat}{RAT}{Radio Access Technology}
\newacronym{rcn}{RCN}{Research Coordination Network}
\newacronym{rec}{REC}{Radio Edge Cloud}
\newacronym{red}{RED}{Random Early Detection}
\newacronym{renew}{RENEW}{Reconfigurable Eco-system for Next-generation End-to-end Wireless}
\newacronym{re}{RE}{Resource Element}
\newacronym{rf}{RF}{Radio Frequency}
\newacronym{rfc}{RFC}{Request for Comments}
\newacronym{rfr}{RFR}{Random Forest Regressor}
\newacronym{ric}{RIC}{\gls{ran} Intelligent Controller}
\newacronym{rlc}{RLC}{Radio Link Control}
\newacronym{rlf}{RLF}{Radio Link Failure}
\newacronym{rlnc}{RLNC}{Random Linear Network Coding}
\newacronym{rmse}{RMSE}{Root Mean Squared Error}
\newacronym{rnis}{RNIS}{Radio Network Information Service}
\newacronym{rr}{RR}{Round Robin}
\newacronym{rrc}{RRC}{Radio Resource Control}
\newacronym{rrm}{RRM}{Radio Resource Management}
\newacronym{rru}{RRU}{Remote Radio Unit}
\newacronym{rs}{RS}{Remote Server}
\newacronym{rsrp}{RSRP}{Reference Signal Received Power}
\newacronym{rsrq}{RSRQ}{Reference Signal Received Quality}
\newacronym{rss}{RSS}{Received Signal Strength}
\newacronym{rssi}{RSSI}{Received Signal Strength Indicator}
\newacronym{rtt}{RTT}{Round Trip Time}
\newacronym{ru}{RU}{Radio Unit}
\newacronym{rw}{RW}{Receive Window}
\newacronym{rx}{RX}{Receiver}
\newacronym{s1ap}{S1AP}{S1 Application Protocol}
\newacronym{sa}{SA}{standalone}
\newacronym{sack}{SACK}{Selective Acknowledgment}
\newacronym{sap}{SAP}{Service Access Point}
\newacronym{sc2}{SC2}{Spectrum Collaboration Challenge}
\newacronym{scef}{SCEF}{Service Capability Exposure Function}
\newacronym{sch}{SCH}{Secondary Cell Handover}
\newacronym{scoot}{SCOOT}{Split Cycle Offset Optimization Technique}
\newacronym{sctp}{SCTP}{Stream Control Transmission Protocol}
\newacronym{sdap}{SDAP}{Service Data Adaptation Protocol}
\newacronym{sdk}{SDK}{Software Development Kit}
\newacronym{sdm}{SDM}{Space Division Multiplexing}
\newacronym{sdma}{SDMA}{Spatial Division Multiple Access}
\newacronym{sdn}{SDN}{Software-defined Networking}
\newacronym{sdr}{SDR}{Software-defined Radio}
\newacronym{seba}{SEBA}{SDN-Enabled Broadband Access}
\newacronym{sgsn}{SGSN}{Serving GPRS Support Node}
\newacronym{sgw}{SGW}{Service Gateway}
\newacronym{si}{SI}{Study Item}
\newacronym{sib}{SIB}{Secondary Information Block}
\newacronym{sic}{SIC}{Successive Interference Cancellation}
\newacronym{sinr}{SINR}{Signal to Interference plus Noise Ratio}
\newacronym{sip}{SIP}{Session Initiation Protocol}
\newacronym{siso}{SISO}{Single Input, Single Output}
\newacronym{sla}{SLA}{Service Level Agreement}
\newacronym{sm}{SM}{Saturation Mode}
\newacronym{smf}{SMF}{Session Management Function}
\newacronym{smo}{SMO}{Service Management and Orchestration}
\newacronym{sms}{SMS}{Short Message Service}
\newacronym{smsgmsc}{SMS-GMSC}{\gls{sms}-Gateway}
\newacronym{snr}{SNR}{Signal-to-Noise-Ratio}
\newacronym{son}{SON}{Self-Organizing Network}
\newacronym{sptcp}{SPTCP}{Single Path TCP}
\newacronym{srb}{SRB}{Service Radio Bearer}
\newacronym{srn}{SRN}{Standard Radio Node}
\newacronym{srs}{SRS}{Sounding Reference Signal}
\newacronym{ss}{SS}{Synchronization Signal}
\newacronym{ssb}{SSB}{Synchronization Signal Block}
\newacronym{sss}{SSS}{Secondary Synchronization Signal}
\newacronym{st}{ST}{Spanning Tree}
\newacronym{svc}{SVC}{Scalable Video Coding}
\newacronym{ta}{TA}{Timing Advance}
\newacronym{tb}{TB}{Transport Block}
\newacronym{tcp}{TCP}{Transmission Control Protocol}
\newacronym{tdd}{TDD}{Time Division Duplexing}
\newacronym{tdm}{TDM}{Time Division Multiplexing}
\newacronym{tdma}{TDMA}{Time Division Multiple Access}
\newacronym{tfl}{TfL}{Transport for London}
\newacronym{tfrc}{TFRC}{TCP-Friendly Rate Control}
\newacronym{tft}{TFT}{Traffic Flow Template}
\newacronym{tgen}{TGEN}{Traffic Generator}
\newacronym{tip}{TIP}{Telecom Infra Project}
\newacronym{tm}{TM}{Transparent Mode}
\newacronym{to}{TO}{Telco Operator}
\newacronym{tr}{TR}{Technical Report}
\newacronym{trp}{TRP}{Transmitter Receiver Pair}
\newacronym{ts}{TS}{Technical Specification}
\newacronym{tti}{TTI}{Transmission Time Interval}
\newacronym{ttt}{TTT}{Time-to-Trigger}
\newacronym{tx}{TX}{Transmitter}
\newacronym{uas}{UAS}{Unmanned Aerial System}
\newacronym{uav}{UAV}{Unmanned Aerial Vehicle}
\newacronym{udm}{UDM}{Unified Data Management}
\newacronym{udp}{UDP}{User Datagram Protocol}
\newacronym{udr}{UDR}{Unified Data Repository}
\newacronym{ue}{UE}{User Equipment}
\newacronym{uhd}{UHD}{\gls{usrp} Hardware Driver}
\newacronym{ul}{UL}{Uplink}
\newacronym{ultdoa}{UL-TDoA}{Uplink Time Difference of Arrival}
\newacronym{um}{UM}{Unacknowledged Mode}
\newacronym{uml}{UML}{Unified Modeling Language}
\newacronym{upa}{UPA}{Uniform Planar Array}
\newacronym{upf}{UPF}{User Plane Function}
\newacronym{urllc}{URLLC}{Ultra Reliable and Low Latency Communication}
\newacronym{usa}{U.S.}{United States}
\newacronym{usim}{USIM}{Universal Subscriber Identity Module}
\newacronym{usrp}{USRP}{Universal Software Radio Peripheral}
\newacronym{utc}{UTC}{Urban Traffic Control}
\newacronym{vim}{VIM}{Virtualization Infrastructure Manager}
\newacronym{vm}{VM}{Virtual Machine}
\newacronym{vnf}{VNF}{Virtual Network Function}
\newacronym{volte}{VoLTE}{Voice over \gls{lte}}
\newacronym{voltha}{VOLTHA}{Virtual OLT HArdware Abstraction}
\newacronym{vr}{VR}{Virtual Reality}
\newacronym{vran}{vRAN}{Virtualized \gls{ran}}
\newacronym{vss}{VSS}{Video Streaming Server}
\newacronym{wbf}{WBF}{Wired Bias Function}
\newacronym{wf}{WF}{Wired-first}
\newacronym{wlan}{WLAN}{Wireless Local Area Network}
\newacronym{osm}{OSM}{Open Source \gls{nfv} Management and Orchestration}
\newacronym{pnf}{PNF}{Physical Network Function}
\newacronym{drl}{DRL}{Deep Reinforcement Learning}
\newacronym{mtc}{MTC}{Machine-type Communications}

\newacronym{cif}{CI}{cyberinfrastructure}
\newacronym{sonic}{SONiC}{Software for Open Networking in the Cloud}
\newacronym{ocp}{OCP}{Open Compute Project}
\newacronym{snmp}{SNMP}{Simple Network Management Protocol}
\newacronym{raid}{RAID}{redundant array of independent disks}
\newacronym{nfs}{NFS}{Network File Storage}
\newacronym{ci}{CI}{Continuous Integration}
\newacronym{cd}{CD}{Continuous Deployment}
\newacronym{dtn}{DTN}{Data Transfer Node}

\newacronym{dns}{DNS}{Domain Name Service}
\newacronym{nrpe}{NRPE}{Nagios Remote Plugin Executor}
\newacronym{ldap}{LDAP}{Lightweight Directory Access Protocol}
\newacronym{lan}{LAN}{Local Area Network}
\newacronym{vlan}{VLAN}{Virtual LAN}

\newacronym{ipmi}{IPMI}{Intelligent Platform Management Interface}
\newacronym{tor}{ToR}{Top-of-the-Rack}
\newacronym{lmn}{LMN}{Large Memory Node}
\newacronym{bgp}{BGP}{Border Gateway Protocol}
\newacronym{dhcp}{DHCP}{Dynamic Host Configuration Protocol}
\newacronym{vrf}{VRF}{Virtual Routing and Forwarding}
\newacronym{vpn}{VPN}{Virtual Private Network}
\newacronym{rma}{RMA}{Return Merchandise Authorization}
\newacronym{hpc}{HPC}{High Performance Compute}

\newacronym{nu}{NU}{Northeastern University}
\newacronym{asic}{ASIC}{Application-specific Integrated Circuit}
\newacronym{rdma}{RDMA}{Remote Direct Memory Access}
\newacronym{roce}{RoCE}{RDMA over Converged Ethernet}
\newacronym{ovs}{OVS}{Open vSwitch}
\newacronym{frr}{FRR}{Free Range Routing}
\newacronym{ups}{UPS}{Uninterruptible Power Supply}

\newacronym{ntia}{NTIA}{National Telecommunications and Information Administration}
\newacronym{pii}{PII}{Personal and Identifiable Information}
\newacronym{irb}{IRB}{Institutional Review Board}
\newacronym{doi}{DOI}{Digital Object Identifier}

\newacronym{sdo}{SDO}{Standard-Development Organization}
\newacronym{ose}{OSE}{Open Source Ecosystem}
\newacronym{osc}{OSC}{O-RAN Software Community}
\newacronym{dop}{DOP}{Director of Operations}
\newacronym{pm}{PM}{Program Manager}
\newacronym{excom}{EXCOM}{Executive Committee}
\newacronym{iiot}{IIoT}{Industrial \gls{iot}}
\newacronym{lf}{LF}{Linux Foundation}

\newacronym{wiot}{WIoT}{Institute for the Wireless Internet of Things}

\newacronym{otic}{OTIC}{Open Testing \& Integration Centre}

\newacronym{nofo}{NOFO}{Notice of Funding Opportunity}

\newacronym{onr}{ONR}{Office of Naval Research}
\newacronym{afosr}{AFOSR}{Air Force Office of Scientific Research}
\newacronym{afrl}{AFRL}{Air Force Research Laboratory}
\newacronym{arl}{ARL}{Army Research Laboratory}

\newacronym{arc}{ARC}{Aerial Research Cloud}

\newacronym{mno}{MNO}{Mobile Network Operator}
\newacronym{ct}{CT}{Continuous Testing}
\newacronym{oci}{OCI}{Open Container Initiative}
\newacronym{macsec}{MACsec}{Media Access Control Security}
\newacronym{pt}{PT}{Plain Text}
\newacronym{cuda}{CUDA}{Compute Unified Device Architecture}
\newacronym{cbrs}{CBRS}{Citizen Broadband Radio Service}
\newacronym{sas}{SAS}{Spectrum Access System}
\newacronym{rfi}{RFI}{Radio-Frequency Interference}
\newacronym{pal}{PAL}{Priority Access License}
\newacronym{gaa}{GAA}{General Authorized Access}
\newacronym{esc}{ESC}{Environmental Sensing Capability}
\newacronym{ota}{OTA}{Over-the-Air}

\usepackage{bm,cite}
\usepackage{amsmath}
\usepackage{amsbsy}
\usepackage{latexsym}
\usepackage{amssymb}
\usepackage{wasysym}
\usepackage{mathtools}
\usepackage{amsthm}

\newtheorem{proposition}{Proposition}

\DeclareMathAlphabet{\mathbit}{OML}{cmr}{bx}{it}
\DeclareMathAlphabet{\mathsf}{OT1}{cmss}{m}{n}
\DeclareMathAlphabet{\mathTXf}{OT1}{cmss}{bx}{it}

\DeclareMathOperator{\Exp}{E}

\DeclareMathOperator{\SNR}{SNR}
\DeclareMathAlphabet{\mathpzc}{OT1}{pzc}{m}{it}

\newcommand{\E}{{\text{E}}}

\newcommand{\MSE}{{\text{MSE}}}

\title{Bistatic Sensing in 5G NR}

\author{
\IEEEauthorblockN{Rajeev Gangula, Sakthivel Velumani, and Tommaso Melodia 
}
\IEEEauthorblockN{Institute for the Wireless Internet of Things, Northeastern University, Boston, USA
}
}

\usepackage{tikzpagenodes,etoolbox}
\usetikzlibrary{calc}
\usepackage[contents={}]{background}
\AddEverypageHook{%
\ifnumequal{\thepage}{1}{%
    \tikz[remember picture,overlay]{%
        \node[draw,
        minimum width=1.03\textwidth,
        text width=1.02\textwidth,
        font=\scriptsize
        ]
        at ($(current page header area) - (0,5pt)$)
        {%
        This work has been submitted to the IEEE for possible publication.\\
        Copyright may be transferred without notice, after which this version may no longer be accessible.
        };
    }%
}{}
}

\begin{document}


\maketitle

\begin{abstract}
In this work, we propose and evaluate the performance of a \gls{5g} \gls{nr} bistatic \gls{isac} system.
Unlike the full-duplex monostatic \gls{isac} systems, the bistatic approach enables sensing in the current cellular networks without significantly modifying the transceiver design.
The sensing utilizes data
channels, such as the \gls{pusch}, which carries information on the air interface. 
We provide the maximum likelihood estimator
for the delay and Doppler parameters and 
derive a lower
bound on the \gls{mse} for a single target scenario.
Link-level simulations show that it is possible
to achieve significant throughput while accurately estimating the sensing parameters with \gls{pusch}.
Moreover, the results reveal an interesting
tradeoff between the number of reference symbols,
sensing performance, and throughput in the proposed \gls{5g} \gls{nr} bistatic \gls{isac} system.
\end{abstract} 

\blfootnote{This work is supported by OUSD (R\&E) through Army Research Laboratory Cooperative Agreement Number W911NF-24-2-0065. The views and conclusions contained in this document are those of the authors and should not be interpreted as representing the official policies, either expressed or implied, of the Army Research Laboratory or the U.S. Government. The U.S. Government is authorized to reproduce and distribute reprints for Government purposes notwithstanding any copyright notation herein.}
    \section{Introduction} \label{sec:intro}

The radio spectrum is a valuable and finite natural resource.
Many technologies that have become essential in our societies, such as radio and television broadcasting, mobile networks, aviation, satellites, radar, and defense services, rely on it.
Historically, communication and radar networks are designed independently and operate on
separate dedicated bands of frequencies.
However, with future wireless networks facing severe spectrum congestion,
the need for spectrum sharing and/or jointly designing communication and sensing (radar) networks has emerged. In this regard, \gls{isac} is expected to become fundamental to next-generation wireless system design.
Initial steps towards standardization in terms of
use cases, requirements, and channel models have already been taken in the \gls{3gpp} \cite{3gpp22837}.

In an \gls{isac} system, communication and sensing functionalities can be performed over a common network infrastructure and radio spectrum through cooperation and/or joint design \cite{Liuisactut2022,gangula2024listen,luo2024integrated}.
While much of the \gls{isac} research is concentrated on waveform design, beamforming, and resource allocation \cite{luo2024integrated,zhou2022integrated}, existing cellular systems already possess some environmental sensing capabilities \cite{kanharetarget2021,evers2024analysis,wei2024multiple,Madownlink2022,khosroshahi2024leveraging,wei5gprs2023,KhosroshahiDoppler2024,NatarajaIsac2024,tapiobistatic2024,strumamulti2010,braunthesis}.
The \gls{ofdm} physical layer waveform used in \gls{5g} \gls{nr} is shown to offer comparable sensing performance in terms of range and Doppler resolution to that of 
traditional radar chirp waveforms \cite{strumamulti2010,braunthesis,fink2015comparision,Gaudio2019ofdmotfs}.

In a cellular \gls{isac} system, sensing functionality can be incorporated into the \gls{gnb} and \gls{ue}.
In a monostatic configuration, a node (UE or gNB) relies on the echo of its transmitted signal to perform sensing. The node can operate in half-duplex or full-duplex mode. While a full-duplex system offers better sensing performance and spectrum utilization, it requires receivers with self-interference cancellation capabilities. In the bistatic scenario, the  
transmitter and receiver are physically separated. 
The gNB and UE can perform sensing based on the signal received in the \gls{ul} and \gls{dl}, respectively.

Several works have examined \gls{isac}-enabled \gls{5g} systems
\cite{Madownlink2022,wei5gprs2023,khosroshahi2024leveraging,KhosroshahiDoppler2024,NatarajaIsac2024,tapiobistatic2024,Lipeformnace2022}.
While the authors in \cite{Madownlink2022,wei5gprs2023,khosroshahi2024leveraging} consider full-duplex monostatic scenario, the bistatic setting is considered in \cite{KhosroshahiDoppler2024,NatarajaIsac2024,tapiobistatic2024,Lipeformnace2022}. 
The majority of the bistatic \gls{isac} works in \gls{5g} \gls{nr} 
rely on reference signals such as \gls{prs}, \gls{dmrs} and \gls{csi-rs} for sensing \cite{KhosroshahiDoppler2024,NatarajaIsac2024,tapiobistatic2024}.
However, the resources allocated for these signals are negligible compared to data symbols.
Therefore, unlike the monostatic scenario where a node can use both reference and data symbols for sensing, the sensing performance of bistatic systems in \cite{KhosroshahiDoppler2024,NatarajaIsac2024,tapiobistatic2024} tends to suffer as they depend only on reference signals.

On the other hand, joint sensing and data detection methods that rely on both reference and data symbols in a bistatic \gls{isac} setting is considered in \cite{Lipeformnace2022,zhao2024joint}.
While the authors in \cite{Lipeformnace2022} rely on \gls{sic} technique to remove the \gls{los} path from the received signal and then perform sensing, joint iterative channel parameter estimation (sensing) and data detection is used in \cite{zhao2024joint}. However, none of them assume any channel coding and rate matching in their works, which are essential in any practical wireless system.

In \gls{5g} \gls{nr}, the data packet or \gls{tb} is transmitted over the \gls{pdsch} in the \gls{dl} and \gls{pusch} in the \gls{ul}. If the receiver (\gls{ue} in the \gls{dl} or \gls{gnb} in \gls{ul}) successfully decodes the \gls{tb}, it can reconstruct the modulated data symbols and then use both data and reference symbols for sensing, similar to a monostatic system.
However, successful decoding of a \gls{tb} relies on many factors such as \gls{snr}, underlying wireless channel and its estimate, coding gain, etc,.
While a large number of pilots can increase the channel estimate quality and reduce the \gls{bler}, the overhead results in a lower rate.
Therefore, it is important to validate the bistatic \gls{isac} \gls{5g} system performance.

In this paper, we study the sensing and throughput performance of the \gls{pusch} in a bistatic scenario.
The sensing relies on both data and pilot symbols. Specifically, the contributions are summarized as follows:

\begin{itemize}
    \item We propose an ISaC receiver which relies on the decoded \gls{5g} \gls{nr} \gls{pusch} data and \gls{dmrs} for sensing.
    \item A maximum likelihood estimator for the delay and Doppler parameters is presented.
    \item Using \gls{crlb} analysis, we provide a bound on the \gls{mse} in a single target scenario. 
    \item The performance of \gls{pusch} is evaluated in terms of sensing \gls{mse} and throughput by numerical evaluations.
\end{itemize}


\section{System Model} \label{sec:sysmod}

We consider a \gls{5g} \gls{nr} bistatic \gls{isac} system with a single antenna \gls{ue} and \gls{gnb}.
The \gls{gnb} and \gls{ue} are located at known coordinates
in a two-dimensional plane.
While the proposed framework can be applied in \gls{dl} and \gls{ul}, without loss of generality, we focus on the \gls{ul} scenario.
The transmitted signal by the \gls{ue} arrives
at the \gls{gnb} via the \gls{los} and $P-1$ reflected paths.
We assume a single-bounce multipath scenario where
each point target causes a reflected path.
All targets lie on the same plane where the \gls{ue} and \gls{gnb} are located.
The considered bistatic geometry is shown in Fig \ref{fig:ch_model}.

\subsection{Transmit Chain}\label{sec:txchain}
The communication and sensing is performed based on the \gls{pusch} transmitted over an \gls{ofdm} slot consisting of $L$ symbols in the time domain and $K$ sub-carriers in the frequency domain.
The OFDM symbol duration is given by
$T_s = T +T_{\text{cp}}$, where $T$ and $T_{\text{cp}}$ represent the data and \gls{cp} duration, respectively.
The transmitted baseband signal by the \gls{ue} in this slot is given by

\begin{equation}\label{eq:tx_baseband}
	{x}\left( {t} \right) = \sum\limits_{{{l}} = 0}^{L - 1} {\sum\limits_{k = 0}^{K - 1} x_{k,l} e^ {j2\pi f_k (t-T_{\text{cp}}-lT_s) } }
    {{\rm{rect}}}\left( {{t - l{T_s}}}\right),
\end{equation}
where $f_k=k \Delta f$, is the k-th subcarrier with subcarrier spacing $\Delta f = 1/T$, $x_{k,l}\in \mathbb{C}$ is the symbol transmitted on the $(k,l)$-th \gls{re} ($l$-th symbol and $k$-th subcarrier), and $\rm{rect}(t)$ is one for $t \in [0, T_s ]$ and 0 otherwise.

\subsection{Channel}\label{sec:chan}
The baseband continuous time wireless channel between the \gls{ue} and \gls{gnb} 
is given by
\begin{align}\label{eq:ch_t}
    {h}(t,\tau) = \sum\limits_{p = 0}^{P - 1} \alpha_p \delta(\tau-\tau_p) e^{j2\pi \nu_{p} t},
\end{align}
where $\alpha_p$ denotes the complex channel coefficient,
$\tau_p$ is the delay, and $\nu_p$ represent the Doppler frequency shift of the $p$-th path respectively.
The Dirac-Delta function is represented by $\delta(.)$.
While $p=0$ represents the \gls{los} path, paths with $p\in[1,P-1]$ represent single bounce reflected paths from the point targets.
We assume that the \gls{gnb} and \gls{ue} are static, hence $\nu_0 = 0$. The maximum delay $\tau_{\text{max}} = \max\{\tau_0,\ldots,\tau_{P-1}\}$ and maximum Doppler shift
$\nu_{\text{max}} = \max\{\nu_0,\ldots,\nu_{P-1}\}$ of the channel are assumed to satisfy
$
\tau_{\text{max}} < T, ~\nu_{\text{max}} < \Delta f.
$

\begin{figure}[t]
\centerline{\includegraphics[width=0.53\textwidth]{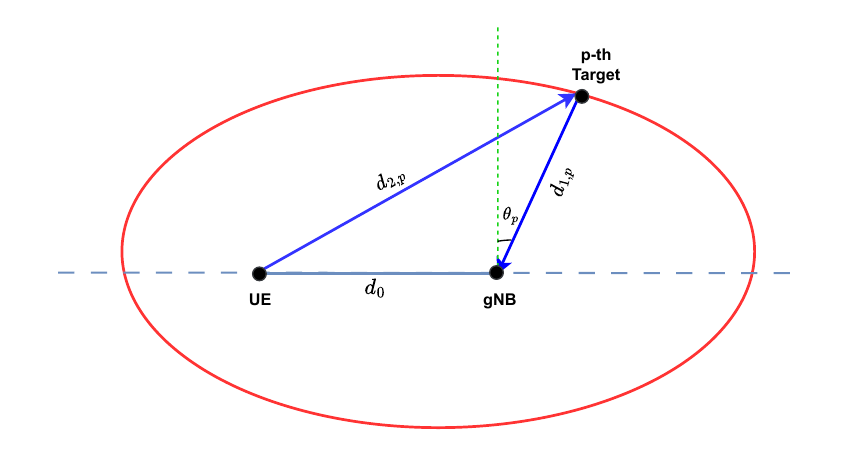}}
\caption{Bistatic geometry with the $p$-th target}
\label{fig:ch_model}
\end{figure}

\subsection{Receiver}\label{sec:rxchain}
Based on the transmitted signal \eqref{eq:tx_baseband} and the channel \eqref{eq:ch_t}, the received signal at the \gls{gnb} can be written as
\begin{equation}\label{eq:rx_signal_t}
    y(t) = \sum\limits_{p = 0}^{P - 1} \alpha_p x(t-\tau_p) e^{j2\pi \nu_{p} t}
    +{w}(t),
\end{equation}
where ${w}(t)\in\mathbb{C}$ represents the circularly
symmetric complex \gls{awgn}.
Assuming a conventional \gls{ofdm} receiver and \gls{isi} free signal
($\nu_{\text{max}} \ll \Delta f$), the received signal
after sampling and \gls{fft} at the \gls{gnb} is given by
\begin{equation}\label{eq:rx_sig_grid}
	{y}_{k,l} = h_{k,l} x_{k,l}
    +{w}_{k,l}
\end{equation}
where 
\begin{equation}\label{eq:chan_rxgrid}
	h_{k,l} = \sum\limits_{{{p}} = 0}^{P - 1} \alpha_p 
    e^{-j2\pi f_{k}\tau_p}
    e^{j2\pi l T_s \nu_p}
\end{equation}
where $w_{k,l}\sim\mathcal{N}(0,\sigma^2)$ represents the complex \gls{awgn} at the $(k,l)$-th \gls{re} with $k\in[0,K-1],~l\in[0,L-1]$.
The path amplitudes are normalized, therefore $\sum_{{{p}} = 0}^{P - 1} |\alpha_p|^2 = 1$.
The path \gls{snr} is given by $\SNR_p \triangleq |\alpha_p|^2/\sigma^2$.

\subsection{Target Localization}\label{sec:rxchain}

To find the position and velocity of the target, the \gls{gnb} has to
derive delay and Doppler parameters from an estimate of the \gls{ul} \gls{cir} in \eqref{eq:chan_rxgrid}. The gNB can measure the time delay
between the direct signal (delay in LoS path) and the  
path reflected off the $p$-th target as
\begin{align} \label{eq:bistat_tdoa}
    \Delta \tau_p &= \tau_p - \tau_0, \\\label{eq:bistat_1}
    & = (d_{1,p}+d_{2,p})/{c}-{d_0}/{c}, 
\end{align}
where $d_{1,p}$ is the distance of the $p$-th target from the \gls{gnb} and, $d_{2,p}$ is the distance of the $p$-th target from the \gls{ue} , $d_0$ is the LoS distance between the \gls{ue} and \gls{gnb}, and $c$ is the speed of light.
The bistatic geometry is shown in Fig \ref{fig:ch_model}.
Equation \eqref{eq:bistat_1} defines an ellipse, with the gNB and UE as the two focal points.
If the gNB is able to measure the \gls{aoa} $\theta_p$ of the reflected path, the target range $d_{1,p}$ and the relative velocity $v_p$ can be obtained by \cite{greco2011cramer}
\begin{equation}
    d_{1,p} = \frac{d_p^2 - d_0^2}{2(d_p+ d_0 \sin \theta_p)},
\end{equation}
and
\begin{equation}
    \nu_p = \frac{2f_c}{c}v_p
    \sqrt{\frac{1}{2}+
    \frac{d_{1,p}+d_0 \sin \theta_p}{2\sqrt{d_{1,p}^2+d_0^2+2d_{1,p}d_0 \sin\theta_p}}},
\end{equation}
where the multi-path distance
$
 d_p \triangleq d_{1,p}+d_{2,p},
$
$f_c$ is the carrier frequency.
If the \gls{gnb} has perfect knowledge
of the LoS delay $\tau_0$ and \gls{aoa} $\theta_p~,\forall p$,
from delay ($\tau_p$) and Doppler ($\nu_p$) estimates, the gNB can estimate the relative velocity and position of the target.
The \gls{aoa} can be estimated with a \gls{gnb} having multiple antennas. However, it is out of the scope of this work.
In this paper the sensing task consists of estimating the delay and Doppler parameters $(\tau_p,\nu_p)$.


\section{5G NR PUSCH}\label{sec:5gpusch}

The \gls{pusch} is the primary physical uplink channel through
which the \gls{ue} transmits data to the \gls{gnb}.
The data payload in the form of \gls{tb} (from higher layers)
arrives at the \gls{phy} layer of the \gls{ue}.
The \gls{phy} layer performs several operations: \gls{crc} attachment for error detection, channel coding, rate matching, and modulation on the \gls{tb} to obtain \gls{pusch} data symbols.
Depending on the chosen \gls{mcs}, different modulation formats, such as QPSK, 16-QAM, or 64-QAM, can be used. 
The \gls{pusch} is associated with reference or pilot signals in the form of \gls{dmrs}, which are used for channel estimation as part of the coherent \gls{pusch} demodulation process. The
detailed \gls{dmrs} configuration and mapping of \gls{pusch} into the resource grid is
defined in \cite{3gpp2018nr_38_212,3gpp2018nr_38_214}.
The \gls{ofdm} waveform generation is described in Section \ref{sec:txchain}, and the \gls{pusch} transceiver chain is shown in Fig \ref{fig:pusch_isac}.

Moreover, \gls{harq} is used to improve the \gls{bler} in \gls{5g} \gls{nr}.
\gls{harq} allows for multiple retransmissions of a \gls{tb}
before declaring a decoding failure.
Each retransmission is associated with additional redundancy bits, which improves the probability of successful decoding when combined with previous
round data.
A maximum of four retransmissions of a \gls{tb} is allowed.
The details are provided in the \gls{3gpp} standards \cite{3gpp2018nr_38_212}.

\begin{figure*}
\centerline{\includegraphics[width=14.5cm,height=5cm]{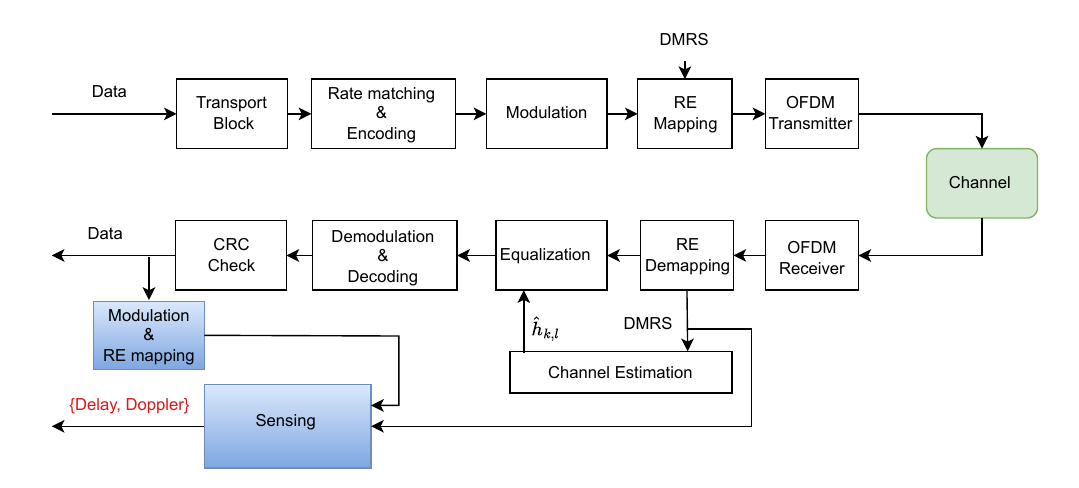}}
\caption{ISaC PUSCH Transceiver}
\label{fig:pusch_isac}
\end{figure*}


\section{ISAC Receiver}\label{sec:isacreceiver}

The goal of the \gls{isac} receiver, shown in Fig. \ref{fig:pusch_isac},
is to decode the \gls{tb} and estimate the delay and Doppler shift parameters.
We now present the throughput of PUSCH followed by the sensing performance analysis.


\subsection{Communication Throughput}\label{sec:commrx}
The PUSCH is scheduled over an \gls{ul} slot consisting of $N_{\text{d}}$ data and $N_\text{p}$ \gls{dmrs} \gls{re}s.
We assume that all \gls{re}s of the slot are used for \gls{pusch}, hence, $N_\text{p}+N_{\text{d}} = KL$.
Let $R_\text{mcs}$ and $Q_\text{mcs}$ denote the 
code rate and modulation order for a given \gls{mcs} index
as defined in \cite[Section~6.1.4]{3gpp2018nr_38_214}.
For a given \gls{mcs}, the average throughput with \gls{harq} is given by \cite{wu2010performance}
\begin{equation}
\bar{R} = \frac{N_{\text{d}} Q_\text{mcs} R_\text{mcs}}{\Exp[X]}  \left(1-\prod_{i=1}^{4}P_i\right)
\text{bits/slot},
\end{equation}
where $P_i$ denotes the TB decoding error probability in the $i$-th, $i\in[1,4]$, \gls{harq} round and $\Exp[X]$ denotes the
expected number of HARQ rounds per TB transmission.
The error probabilities satisfy $0\leq P_i \leq 1,~P_i \leq P_j ~\text{if}~i > j$, 
and $\sum_{i=1}^{4}P_i = 1$.
 Assuming independence among the probabilities, following \cite{wu2010performance} we have
\begin{equation}\label{eq:rounds}
   \Exp[X] = 1+\sum_{i=1}^3 \prod_{j=1}^i P_j 
\end{equation}

\subsection{Sensing}\label{sec:sensing}

The sensing unit estimates the multi-path parameters 
$(\tau_p,\nu_p),~p\in[1,P-1]$ based on the
received signal in \eqref{eq:rx_sig_grid}.
The \gls{los} path delay $\tau_0$ can be used to estimate the position of the \gls{ue}. However, similar to the work in \cite{Lipeformnace2022}, we assume that the 
\gls{ue} position is known, and the \gls{gnb} can successfully remove the \gls{los} path from the \gls{ul} \gls{cir} for sensing.
Joint \gls{ue} positioning and sensing will be addressed in the extended version of this paper. 

The sensing signal processing is related to the decoding of the \gls{pusch} \gls{tb}.
\textbf{Scenario 1:} On a given slot, the gNB is unable to decode the TB, i.e.,
in a \gls{crc} failure situation, it relies only on DMRS REs for sensing.
\textbf{Scenario 2:} 
In a slot where the TB is successfully decoded, the \gls{gnb} can reconstruct the symbols on data REs, and use all REs for sensing.
The sensing unit is illustrated in Fig \ref{fig:pusch_isac}.


For simplicity, we focus on a single target scenario with channel parameters $(\alpha_1,\tau_1,\nu_1)$,
and assume that \gls{los} component has successfully removed from the \gls{cir} for sensing.
If the transmitted symbol $x_{k,l}$ is known to the \gls{isac} receiver, we can undo the phase of the symbol on that RE and obtain
\begin{equation}\label{eq:zkl}
    z_{k,l} = |x_{k,l}| \alpha_1
    e^{-j2\pi f_{k}\tau_1}
    e^{j2\pi l T_s \nu_1}
    + \tilde{w}_{k,l},
\end{equation}
where $|x_{k,l}|$ is the magnitude of the transmitted symbol $x_{k,l}$ and $\tilde{w}_{k,l} \sim \mathcal{N}(0, \sigma^2)$ is circular symmetric
i.i.d. Gaussian noise. We assume a QPSK modulation scheme, hence, $|x_{k,l}|=1,~\forall k,l$.
Note that in \textit{scenario 1} only DMRS is used for sensing, and measurements on non-DMRS REs $z_{k,l} = 0$.

The unknown parameters $\bm{\lambda} = (\alpha_1,\tau_1,\nu_1)$
can be estimated using Maximum Likelihood (ML) estimator.
The ML estimator minimizes the log-likelihood function given by \cite{Gaudio2019ofdmotfs},
\begin{equation}\label{eq:llrs1}
l_1(\bm{Z}|\bm{\lambda}) = 
\underset{(k,l)\in \text{DMRS}}{\sum\sum}
 \left|z_{k,l}-\alpha_1
    e^{j2\pi (l T_s \nu_1 - f_{k}\tau_1)}\right|^2 ,
\end{equation}
\begin{equation}\label{eq:llrs2}
l_2(\bm{Z}|\bm{\lambda}) = 
\sum_{k=0}^{K-1} \sum_{l=0}^{L-1} \left|z_{k,l}-\alpha_1  
    e^{j2\pi (l T_s \nu_1 - f_{k}\tau_1)}\right|^2,
\end{equation}
where $\bm{Z}$ consists $z_{k,l}$ as its elements, \eqref{eq:llrs1} is used in \textit{scenario 1} and \eqref{eq:llrs2}
for \textit{scenario 2}. The ML solution is obtained by computing a two-dimensional periodogram as in \cite{braunthesis,Gaudio2019ofdmotfs}.

\subsection{Fisher Information}\label{sec:crlb} 

Considering the parameter vector $\bm{\lambda} = [h, \phi,\tau_1,\nu_1]$, where $h=|\alpha_1|,~\phi = \angle \alpha_1$, we can rewrite \eqref{eq:zkl} 
as
$$
    z_{k,l} =  s_{k,l}
    + \tilde{w}_{k,l},
$$
where 
$
    s_{k,l} =  h e^{j\phi}
    e^{-j2\pi f_{k}\tau_1}
    e^{j2\pi l T_s \nu_1}
$.
The Fisher information matrix can then be obtained by computing \cite{Gaudio2019ofdmotfs}
\begin{equation}\label{eq:fim_formu}
    {\bf{F}}^m_{i,j}(\bm{\lambda}) = \frac{2}{\sigma^2}   \Re\left\{
    \sum_k\sum_l \left[\frac{\partial s_{k,l}}{\partial \bm{\lambda}_i}\right]^*
    \left[\frac{\partial s_{k,l}}{\partial \bm{\lambda}_j}\right]
    \right\}, m\in\{1,2\},
\end{equation}
and given on the top of the next page, where the indices in summation takes only DMRS REs for \textit{scenario 1}, while all RE indices are considered for \textit{scenario 2}.

\begin{figure*}
\begin{equation}
\small
{\bf{F}}^m
=\frac{2}{\sigma^2}
\begin{bmatrix}\label{eq:fim_matrix}
KL & 0 & 0 & 0\\
0 & h^2KL & -2\pi h^2\Delta f\underset{(k,l)}{\sum \sum} k & 2\pi h^2T_s\underset{(k,l)}{\sum \sum} l\\
0 & -2\pi h^2\Delta f\underset{(k,l)}{\sum \sum} k & (2\pi)^2 h^2{\Delta f}^2\underset{(k,l)}{\sum \sum} k^2 & -(2\pi)^2 h^2\Delta fT_s\underset{(k,l)}{\sum \sum} lk\\
0 & 2\pi h^2T_s\underset{(k,l)}{\sum \sum} l & -(2\pi)^2 h^2\Delta fT_s\underset{(k,l)}{\sum \sum} kl & (2\pi)^2 h^2T_s^2 \underset{(k,l)}{\sum \sum}  l^2\\
\end{bmatrix}
, (k,l) \in \left\{ 
  \begin{array}{ c l }
    \text{DMRS REs} & m = 1 \\
    \text{All REs}  & m = 2
  \end{array}
\right.
\end{equation}
\end{figure*}

\subsection{Estimation Error}\label{sec:esterr} 

Let $\MSE_m(\bm{\lambda_}n)$ represent the \gls{mse} of estimating the $n$-th parameter, $n \in [1,4]$ of $\bm{\lambda} = [h, \phi,\tau_1,\nu_1]$, in the $m$-th, $m\in\{1,2\}$, scenario.
\begin{proposition}
For a given MCS, the PUSCH sensing performance is given by
\begin{equation}\label{eq:mse_exp}
  \MSE(\bm{\lambda_}n) = (1-\rho) \MSE_1(\bm{\lambda}_n)
  + \rho \MSE_2(\bm{\lambda}_n),
\end{equation}
where $\rho = \frac{1-\prod_{i=1}^{4}P_i}{\E[X]}$,
$P_i$ denotes the decoding error probability
in the $i$-th \gls{harq} round, and $\E[X]$ is given in \eqref{eq:rounds}. 
\end{proposition}
\begin{proof}
If the TB is not decoded at the end of the final \gls{harq} round, we are in scenario 1.
This occurs with a probability $\Pi_{i=1}^{4}P_i$, and the \gls{mse} in this case is $\MSE_1(\bm{\lambda}_n)$.
In the complimentary case, with probability ${1-\Pi_{i=1}^{4}P_i}$, let $\E[X]$ denote the expected number of \gls{harq} rounds to decode the TB. 
The average MSE in this case is given by
$$
 \left(1-\frac{1}{\E[X]}\right)\MSE_1(\bm{\lambda}_n)+ \frac{\MSE_2(\bm{\lambda}_n)}{\E[X]}. 
$$
From the above arguments and using the total probability law we obtain \eqref{eq:mse_exp}.
\end{proof}
Since all REs are used in scenario 2, $\MSE_2(\bm\lambda_i)\leq\MSE_1(\bm\lambda_i)$. For an unbiased estimator, 
\begin{equation}\label{eq:crlbexp}
\MSE_m(\bm{\lambda}_n) \geq {({\bf{F}}^m_{n,n})}^{-1},~n \in [1,4], m \in\{1,2\}.
\end{equation}
From \eqref{eq:mse_exp} and \eqref{eq:crlbexp} we can find a lower bound for the \gls{mse} of the proposed \gls{isac} scheme.
Note that we are restricted to MCS index with QPSK modulation in this paper.

\section{Numerical results}\label{sec:result}

In this section, we utilize the MATLAB 5G Toolbox™ \gls{pusch} functions to assess the performance of the proposed \gls{isac} system. The system parameters are listed in Table \ref{tab:1}.
The \gls{mcs} table is taken from \cite{3gpp2018nr_38_214}.
The \gls{pusch} and \gls{dmrs} configuration parameters used in the simulations are
listed in Table \ref{tab:2}.
We can vary the number of DMRS in a slot by changing the 
{\tt DMRS additional position} parameter. It represents the number of extra DMRS in a slot in addition to the default DMRS symbol. 

The system is evaluated at different \gls{snr} values, and for each SNR, 2000 \gls{pusch} transmissions (100 frames) are simulated.
For every iteration, the channel is generated according to
\eqref{eq:chan_rxgrid} with $P=2$ and random sensing parameters.
The LoS and the reflected path off the target have
SNRs $\SNR_0 \triangleq |\alpha_0|^2/\sigma^2$ and  
$\SNR_1 \triangleq |\alpha_1|^2/\sigma^2$, respectively.
For communication, the SNR is defined as $\SNR_c \triangleq \SNR_0+\SNR_1$. For sensing, we assume that the LoS path has been successfully removed, and the measurement is generated as in \eqref{eq:zkl} with $\SNR_1$.
In the simulated scenario, we have $\SNR_c = 10 \SNR_1$. 

For a given \gls{mcs}, the \gls{pusch} ISac performance is 
evaluated in terms of range and Doppler \gls{rmse} for sensing, and average throughput for communication. 
In Fig, \ref{fig:mcs_range} the range ($d_1 = c\tau_1$) RMSE and throughput of the proposed method are plotted for PUSCH MCS 0 with 2 and 4 pilot symbols. The curve with $\text{BLER} = 0$ acts as a lower bound, corresponding to the idealized perfect TB decoding scenario.
Note that the communication SNR is given by $\SNR_c$ while for sensing it is $\SNR_1$.
Fig \ref{fig:mcs_0_1} shows the Doppler RMSE and throughput plots for MCS 0 with two and four DMRS symbols. 
From these figures, it can be seen that while the additional DMRS symbols tend to improve the sensing performance, they result in lower throughput, especially when the SNR is increased. 
In Fig \ref{fig:mcs_comp} we compare the ISaC performance by varying the MCS for fixed DMRS symbols.
One can see the tradeoff relation between the MCS index (coding rate), SNR, sensing RMSE, and throughput.
Finally, Fig \ref{fig:mse_lb} shows the Doppler RMSE and the numerically evaluated lower bound using the analysis in Section \ref{sec:esterr}.

\begin{small}
\begin{table}[htbp]
\caption{System Parameters}
\centering
\begin{tabular*}{\columnwidth}{@{\extracolsep{\fill}}cc}
\toprule
Parameters & Values \\
\midrule
Subcarrier Spacing ($\Delta f$) & 30 KHz \\
Number of PRBs & 106 \\
Subcarriers ($K$)  & 1272 \\
OFDM Symbols in a slot ($L$) & 14 \\
Cyclic Prefix & Normal\\
\bottomrule
\end{tabular*}
\label{tab:1}
\end{table}
\begin{table}[htbp]
\caption{PUSCH and DMRS Parameters}
\centering
\begin{tabular*}{\columnwidth}{@{\extracolsep{\fill}}cc}
\toprule
Parameters & Values \\
\midrule
Mapping & type A \\
Symbol allocation & [0,\ldots,13] \\
PRB allocation & [0,\ldots,105]  \\
Modulation  & QPSK \\
DMRS additional position & \{1,2,3\} \\
DMRS configuration and length & type 1, 1\\
\bottomrule
\end{tabular*}
\label{tab:2}
\end{table}
\end{small}

\begin{figure}
    \centering

    \subfigure
    {
        \includegraphics[width=0.45\textwidth]{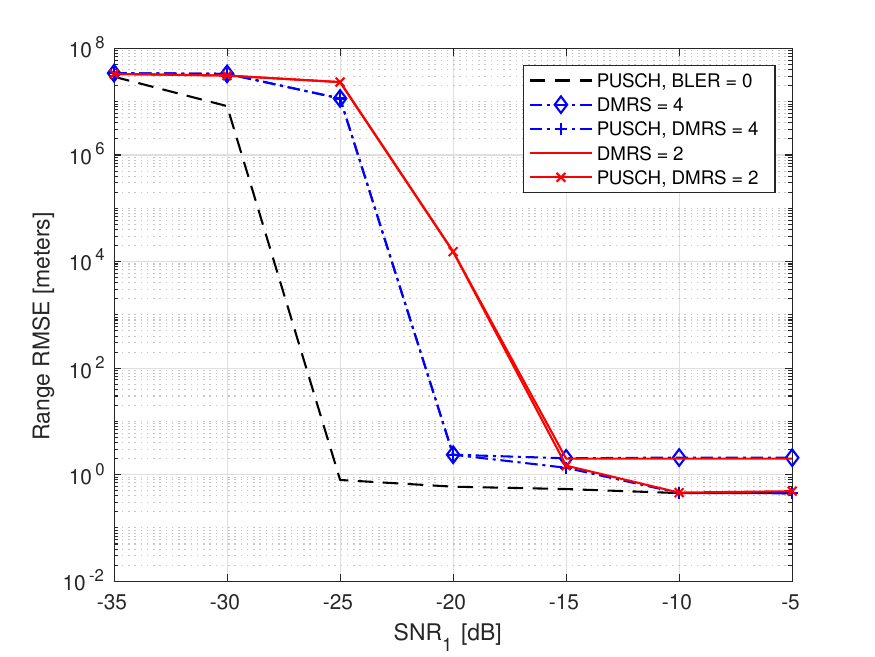}
        \label{fig:mcs_comp_s}
    }
    \subfigure
    {
        \includegraphics[width=0.45\textwidth]{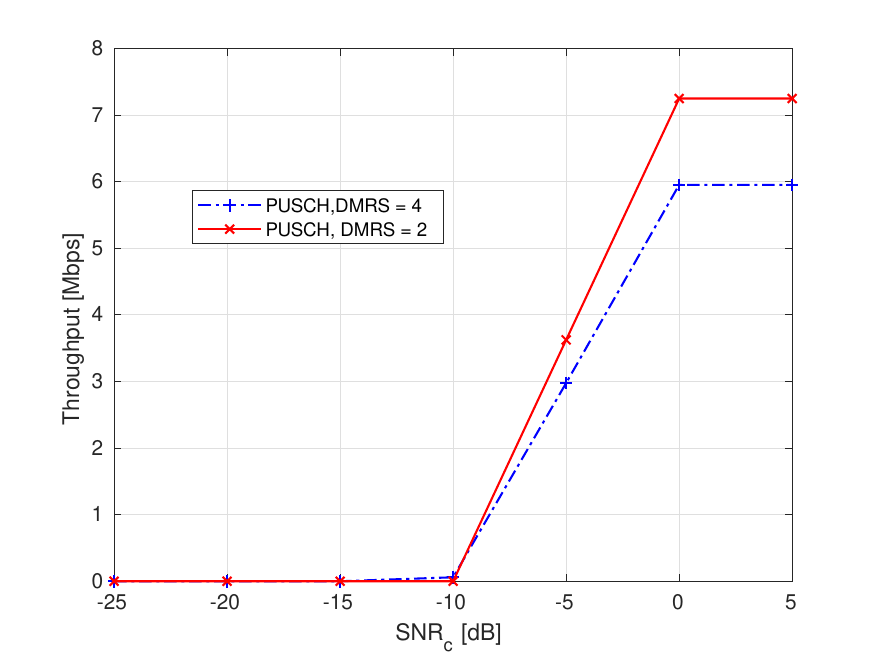}
        \label{fig:mcs_comp_r}
    }
    \caption{PUSCH range RMSE and throughput for MCS 0}
    \label{fig:mcs_range}
\end{figure}

\begin{figure}
    \centering

    \subfigure
    {
        \includegraphics[width=0.45\textwidth]{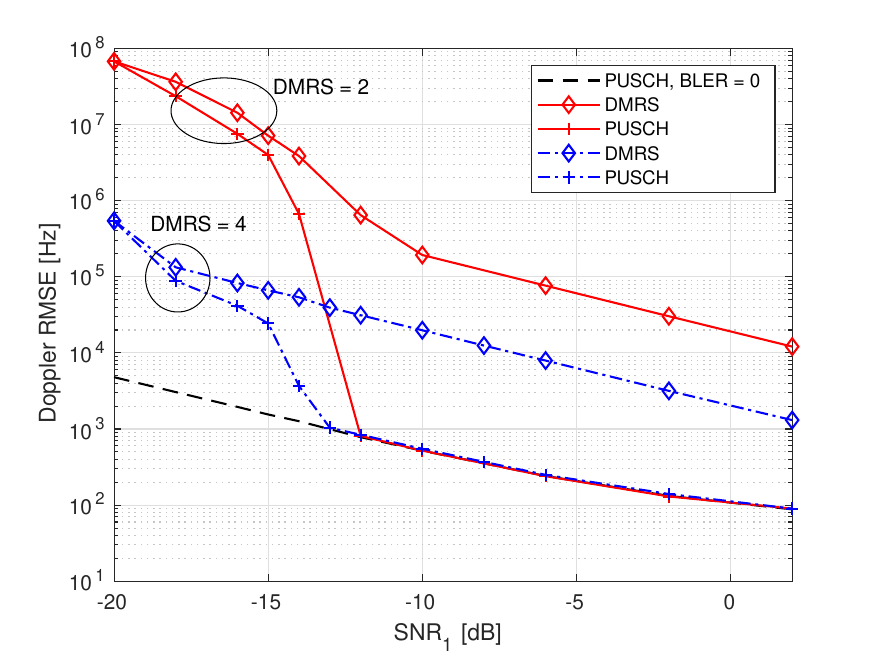}
        \label{fig:mcs_0_1_s}
    }
    \subfigure
    {
        \includegraphics[width=0.45\textwidth]{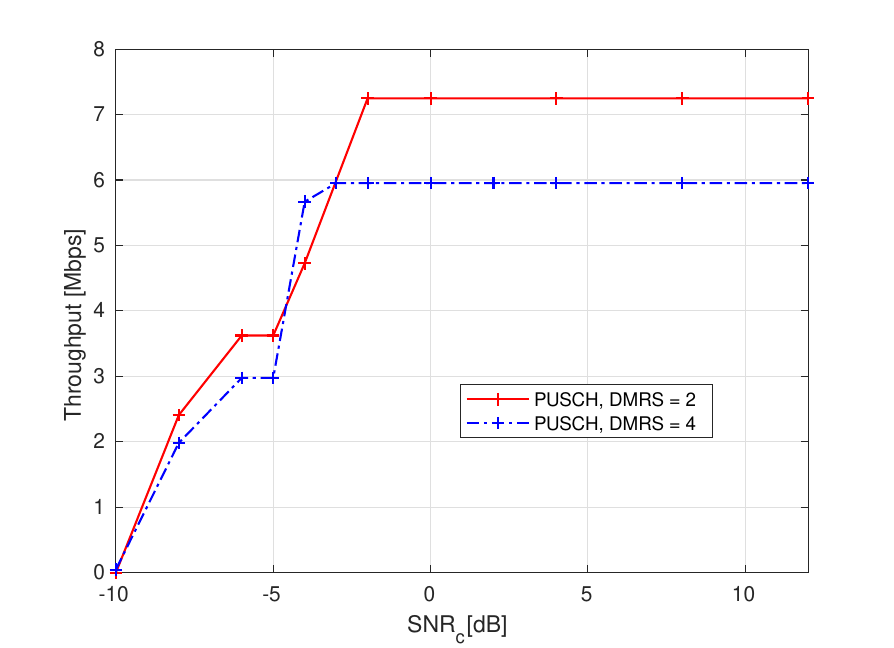}
        \label{fig:mcs_0_1_c}
    }
    \caption{PUSCH Doppler RMSE and throughput for MCS 0}
    \label{fig:mcs_0_1}
\end{figure}

\begin{figure}
    \centering

    \subfigure
    {
        \includegraphics[width=0.45\textwidth]{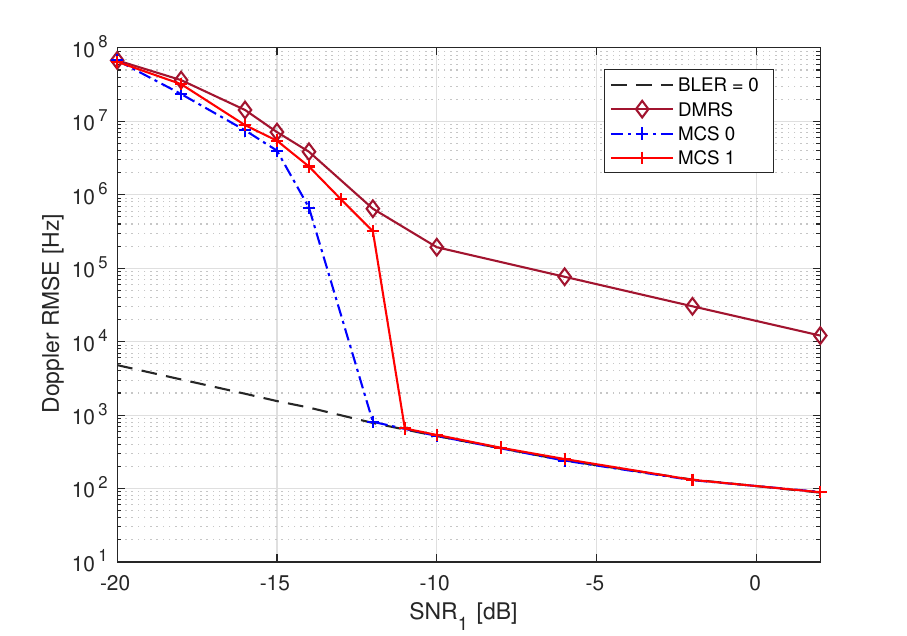}
        \label{fig:mcs_comp_s}
    }
    \subfigure
    {
        \includegraphics[width=0.45\textwidth]{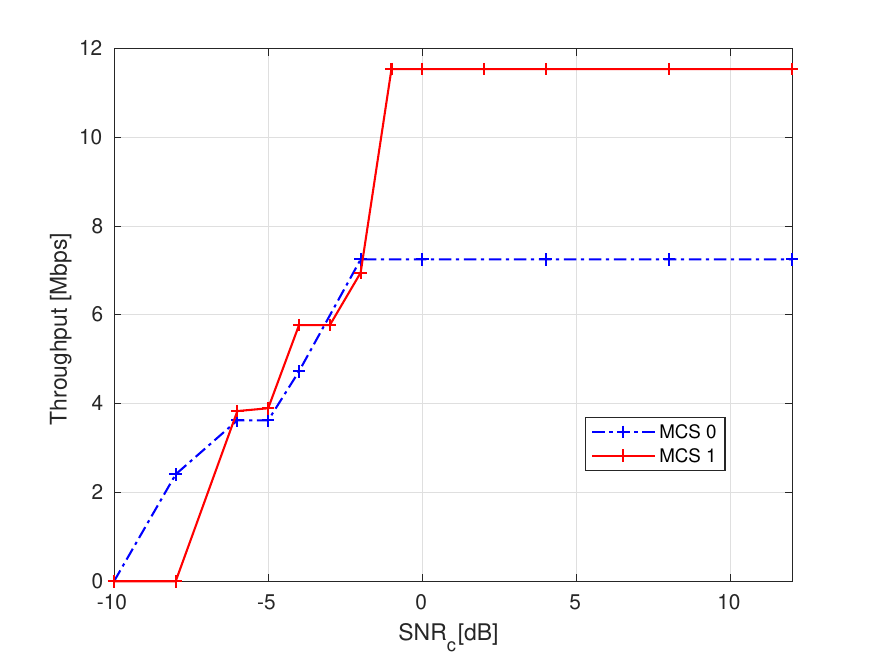}
        \label{fig:mcs_comp_r}
    }
    \caption{PUSCH ISaC performance for MCS 0 and 1}
    \label{fig:mcs_comp}
\end{figure}

\begin{figure}[t]
\centerline{\includegraphics[width=0.4\textwidth]{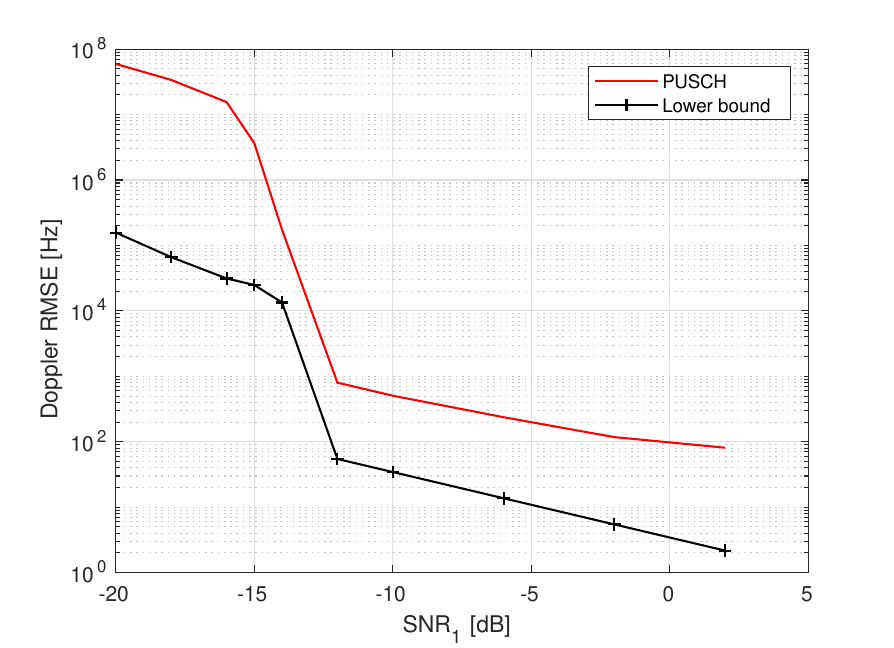}}
\caption{PUSCH Doppler RMSE and its lower bound, MCS 0}
\label{fig:mse_lb}
\end{figure}

\section{Conclusion} \label{sec:concl}
Bistatic ISaC systems can provide sensing functionality in cellular networks without major modifications in the system design and hardware.
We have proposed a PUSCH-based ISaC framework where
sensing is performed based on the decoded PUSCH data and DMRS.
Numerical results show that significant sensing performance can be obtained while not sacrificing the throughput.
Moreover, an interesting tradeoff between coding rate, channel variation, and SNR is observed, which is being evaluated in an extension of this work.

\bibliographystyle{IEEEtran}
\bibliography{References}

\end{document}